\documentclass[11pt,a4paper]{article}
\usepackage[margin=1in]{geometry}

\usepackage{mathptmx} 
\usepackage{microtype} 

\usepackage{amsmath,amsthm,amssymb}
\usepackage{mathtools}
\usepackage{xspace}
\usepackage{url}
\usepackage{hyperref}
\usepackage[dvipsnames]{xcolor}
\usepackage{algpseudocode}
\usepackage{appendix}
\usepackage{babel}
\usepackage[capitalise,noabbrev]{cleveref}
\usepackage{bbm}

\usepackage{multirow}
\usepackage{tablefootnote}
\usepackage{afterpage}
\usepackage{lipsum}
\usepackage{booktabs}
\usepackage{thm-restate}
\usepackage{graphicx}
\usepackage{subcaption}
\usepackage{lipsum}

\newcommand{\abs}[1]{\ensuremath{\mathopen\lvert #1 \mathclose\rvert}}

\newcommand{\NN}{\ensuremath{\mathbb{N}}}

\newcommand{\RR}{\ensuremath{\mathbb{R}}}

\newcommand{\eps}{\varepsilon}
\newcommand{\e}{\mathbb{E}}
\newcommand{\q}{\mathbf{q}}
\newcommand{\p}{\mathbf{p}}
\newcommand{\w}{\mathbf{w}}

\newcommand{\I}{\mathcal{I}}

\newcommand{\G}{\mathcal{G}}
\newcommand{\In}{\mathcal{I}(\mathcal{\phi})}

\newcommand{\NP}{\ensuremath{\mathrm{NP}}}

\newcommand{\PPAD}{\ensuremath{\mathrm{PPAD}}}
\newcommand{\LinearFIXP}{\ensuremath{\mathrm{Linear\text{-}FIXP}}}
\newcommand{\FIXP}{\ensuremath{\mathrm{FIXP}}}

\usepackage{booktabs}
\let\oldtabular\tabular 
\renewcommand{\tabular}{\normalsize\oldtabular}

\theoremstyle{plain}
\newtheorem{theorem}{Theorem}[section]
\newtheorem{lemma}[theorem]{Lemma}

\newtheorem{proposition}[theorem]{Proposition}

\newtheorem{observation}[theorem]{Observation}

\theoremstyle{definition}
\newtheorem{definition}[theorem]{Definition}

\theoremstyle{remark}

\allowdisplaybreaks

\hypersetup{colorlinks=true,allcolors=magenta}
\title{On the complexity of Pareto-optimal and envy-free lotteries\thanks{IC and KAH were partially supported by the Independent Research Fund Denmark under grants 2032-00185B and 9040-00433B, respectively.}}

\author{Ioannis Caragiannis \and Kristoffer Arnsfelt Hansen \and Nidhi Rathi}
\date{Department of Computer Science, Aarhus University\\
{\AA}bogade 34, 8200 Aarhus N, Denmark\\
Email: \texttt{\{iannis,arnsfelt,nidhi\}@cs.au.dk}}

\begin{document}

\maketitle

\begin{abstract}
We study the classic problem of dividing a collection of indivisible resources in a \emph{fair} and \emph{efficient} manner among a set of agents having varied preferences. \emph{Pareto optimality} is a standard notion of economic efficiency, which states that it should be impossible to find an allocation that improves some agent’s utility without reducing any other’s. On the other hand, a fundamental notion of fairness in resource allocation settings is that of \emph{envy-freeness}, which renders an allocation to be fair if every agent (weakly) prefers her own bundle over that of any other agent's bundle. Unfortunately, an envy-free allocation may not exist if we wish to divide a collection of indivisible items. Introducing randomness is a typical way of circumventing the non-existence of solutions, and therefore, {\em allocation lotteries}, i.e., distributions over allocations have been explored while relaxing the notion of fairness to \emph{ex-ante} envy freeness.

We consider a general fair division setting with $n$ agents and a family of admissible $n$-partitions of an underlying set of items. Every agent is endowed with \emph{partition-based utilities}, which specify her cardinal utility for each bundle of items in every admissible partition. In such fair division instances, Cole and Tao~(2021) have proved that an ex-ante envy-free and Pareto-optimal allocation lottery is always guaranteed to exist. We strengthen their result while examining the computational complexity of the above total problem and establish its membership in the complexity class \PPAD. Furthermore, for instances with a constant number of agents, we develop a polynomial-time algorithm to find an ex-ante envy-free and Pareto-optimal allocation lottery. On the negative side, we prove that maximizing social welfare over ex-ante envy-free and Pareto-optimal allocation lotteries is \NP-hard.
\end{abstract}
\section{Introduction}
Fairly dividing a collection of resources among individuals (often dubbed as agents) with varied preferences forms a key concern in the design of many social institutions. Such problems arise naturally in many real-world scenarios such as assigning computational resources in a cloud com\-pu\-ting environment,  air traffic management, dividing business assets, allocation of radio and television spectrum, course assignments, and so on \cite{AdjustedWinner,etkin2007spectrum,moulin2004fair,vossen2002fair}. The fundamental problem of fair division lies at the interface of economics, social science, mathematics, and computer science, and its formal study dates back about seven decades~\cite{dubins1961cut,steihaus1948problem}. In the last few decades, the area of fair division has witnessed a flourishing flow of research; see \cite{survey2022,brams1996fair,brandt2016handbook} for excellent expositions. 

Traditionally, in early literature, fair division has been studied for a single \emph{divisible} resource, classically known as \emph{fair cake cutting}. Here, each agent specifies her valuations over a unit interval cake via a probability distribution over $[0,1]$ and the problem is to divide the cake among agents in a fair manner. The quintessential notion of fairness in this line of work is that of \emph{envy-freeness}, introduced by Foley \cite{foley1966resource} and Varian \cite{Varian1974Equity}. A cake division is said to be \emph{envy-free} if every agent prefers her own share of the cake over any other agent's share. Stromquist \cite{stromquist1980cut} famously proved that an envy-free cake division (where every agent receives a connected interval of the cake) is always guaranteed to exist, under mild conditions. Later, Su \cite{su1999rental}  developed another existential proof using Sperner's Lemma and established a connection between the notion of envy-freeness and topology. Such strong existential results have arguably placed the notion of envy-freeness as the flagship bearer of fairness in resource allocation settings.

On the other hand, \emph{Pareto optimality} is a standard notion of economic efficiency, which states that it should be impossible to find an allocation that improves some agent’s utility without reducing any other's. Another important notion of (collective) efficiency measure of an allocation is that of \emph{social welfare} \cite{moulin2004fair}  which is the sum of all the utilities derived by agents from their assigned bundle.

The goal of being fair towards the participating agents and achieving collective (economic) efficiency form the two important paradigms of resource allocation problems.   Unfortunately, for an indivisible set of items, an envy-free allocation may not exist. For example, an instance with two agents having positive value for a single item admits no envy-free allocation.

A fair division instance consists of a set $[n] = \{1,2 \dots, n\}$ of $n$ agents and a set $M$ of items. In the most basic setting, every agent $i$ has an additive utility function $u_i \colon 2^M \to \mathbb{R}$ that specifies her cardinal preferences for the items of a given bundle; in particular, $u_i(j) \coloneqq u_i(\{j\})$ denotes agent $i$'s utility for item $j \in M$. We say an \emph{allocation} is a partition of items into $n$ bundles, where every agent is assigned one bundle. The goal of simultaneously achieving fairness and efficiency is challenging for the problem of allocating indivisible items.
Besides the mentioned fact that an envy-free allocation is not guaranteed to exists, in cases where envy-free allocations do exist, envy-freeness may not be compatible with Pareto optimality~\cite{BL16}.

The above discussion suggests that one should consider distributions over allocations (to be referred as \emph{allocation lotteries}) in order to simultaneously achieve fairness and efficiency guarantees. In the random assignment literature in economics, the idea of constructing a fractional allocation and implementing it as a lottery over deterministic allocations was introduced by Hylland and Zeckhauser \cite{HyllandZ1979-HZ}. Introducing randomness is a typical way of circumventing the non-existence of various solution concepts, especially in computational social choice theory \cite{abdulkadirouglu1998random,aziz2019probabilistic,caragiannis2019stable,dougan2016efficiency}. In the process of exploring allocation lotteries, we appropriately relax the notion of fairness to \emph{ex-ante envy-freeness}, which values the random bundles allocated agents in terms of expected utility. Recent works of \cite{aziz2020simultaneously,budish2013designing,caragiannis2021interim,freeman2020best} deals with various computational aspects of allocation lotteries that are fair and efficient for fair division instances with additive utilities.
Observe that, allocation lotteries that are just ex-ante envy-free or just ex-ante Pareto-optimal can be trivially computed in polynomial time. For the former, one can solve a linear program, while for the latter, one can assign each bundle to the agent that has the highest utility for it. That is, these notions of fairness and efficiency are tractable if dealt with individually. Therefore, the important question is to understand the computational complexity of computing allocation lotteries that are simultaneously ex-ante envy-free and Pareto-optimal. In this work, we consider this question for the most general setting of fair division, as detailed in the following section.


%




\subsection{Context and overview of our results}
In this work, we consider a very general fair division setting with $n$ agents and a family of admissible $n$-partitions of an underlying set of items. Every agent is endowed with \emph{partition-based utilities} that specify her cardinal utility for different bundles in every partition. For such a broad class of fair division instances with partition-based utilities, including negative-valued utilities, the recent work of Cole and Tao \cite{cole2021existence} proves that an ex-ante envy-free and Pareto-optimal allocation lottery is always guaranteed to exist. 

Note that, partition-based utilities provide a much broader way of expressing agents' utilities. In particular, it is possible that an agent may value the exact same bundle of items in two distinct partitions at two different values. Or, there may be a certain partition of items that is not favourable or suitable (depending on the context of application), and this generalization allows us to remove unsuitable partitions from the family of admissible partitions, and still, the existence of ex-ante envy-free and Pareto optimal allocation lotteries is guaranteed. 

In this work, we examine the computational complexity of the above total search problem and strengthen the work of Cole and Tao~\cite{cole2021existence}. In particular, we establish that the problem of finding an ex-ante envy-free and Pareto optimal allocation lottery for fair division instances with partition-based utilities belongs to the complexity class \PPAD. This containment result is even interesting for the special case of a single admissible partition. Namely, our \PPAD\ membership result is for the \emph{exact} search problem, of computing a rational valued lottery. This can be contrasted with the lottery provided by the Hylland-Zeckhauser (HZ) pseudo-market. Vazirani and Yannakakis \cite{VaziraniY2021-HZ} gave a simple example with four agents and four goods where the \emph{unique} HZ equilibrium gives an irrational-valued lottery. This fact means that any algorithm for computing a HZ equilibrium exactly must overcome numerical challenges.
Our result on the other hand gives hope for the possibility of developing a \emph{practical} algorithm for computing an exact ex-ante envy-free and Pareto optimal allocation lottery, for instance by an adaptation of Lemke's algorithm~\cite{Lemke1964-Bimatrix}.

For instances with a constant number of agents, we develop a polynomial-time algorithm to compute an exact ex-ante envy-free and Pareto-optimal lottery. On the negative side, we prove that maximizing social welfare over ex-ante envy-free and Pareto optimal allocation lotteries is \NP-hard.

\subsection{Further related work}
Fairness in resource-allocation settings is extensively studied in the economics, mathematics, and computer science literature (see \cite{brams1996fair,brandt2016handbook,moulin2004fair}). As mentioned above, envy-free allocations may not exist for the case of indivisible items. Since envy-freeness is arguably a fundamental notion of fairness, as evident from its importance in fair cake cutting, there has been a significant body of research aimed towards finding ex-ante envy-free allocation lotteries in the indivisible setting. 
The work of Freeman~et~al.~\cite{freeman2020best} addresses the key question of whether ex-ante envy-freeness
can be achieved in combination with \emph{ex-post envy-freeness up to one item}. They settle it positively by designing an efficient algorithm that achieves both properties simultaneously. 
Caragiannis~et~al.~\cite{caragiannis2021interim} explore the \emph{interim allocation lotteries (iEF)} which provide fairness guarantees that lie between ex-post and ex-ante envy-freeness. They develop
polynomial-time algorithms for computing iEF lotteries that maximize various efficiency notions.
 
Budish~et~al.~\cite{budish2013designing} employ a general class of random allocation mechanisms to achieve ex-ante fairness and efficiency in the presence of real-world constraints. Several other works explore fairness and efficiency guarantees of allocation lotteries as well, but for ordinal utilities~\cite{abdulkadirouglu1998random,bogomolnaia2001new,chen2002improving}.

Another line of research has explored various relaxations of envy-freeness. The notion of \emph{envy-freeness up to one item (EF1)} was introduced by Budish~et~al.~\cite{budish2011combinatorial} as one of the first `good' relaxations of envy-freeness in the indivisible setting. We say an
allocation is EF1 when every agent (weakly) prefers her own bundle over any other agent~$j$'s bundle after removing some item from~$j$’s bundle. EF1 allocations are guaranteed to always exist for general monotone valuations and can be computed efficiently \cite{lipton2004approximately}. Moreover, this fairness notion is compatible with the economic efficiency objective of Pareto-optimality~\cite{CKMPSW19}. Later, \emph{envy-freeness up to any item (EFX)} was introduced by Caragiannis~et~al.~\cite{CKMPSW19} as a refinement of EF1 and is now considered as the most compelling fairness criterion while dividing indivisible items. We say an allocation is EFX when every agent
(weakly) prefers her own bundle than any other agent~$j$’s bundle after removing her least positively-valued item from~$j$’s bundle. Recent works \cite{amanatidis2021maximum,amanatidis2020multiple,chaudhury2020efx,chaudhury2021little} have shown existential guarantees for EFX in various special cases.



\section{The model}

Consider a set $[n]=\{1,2, \dots, n\}$ of $n$ agents and a collection $\mathcal{P}=\{P^1,P^2,\dots, P^m\}$ of admissible partitions of a set $M$ of items. Every partition $P^k$ for $k \in [m]$ consists of $n$ bundles, i.e.,  $\abs{P^k}=n$  and the union of those bundles is $\bigcup_{A \in P^k}A \subseteq M$. Agents are endowed with utility functions $u_i$'s that specify their cardinal preferences for all bundles in every different partition. In particular, the function $u^k_{ij}$ specifies \emph{partition-based cardinal utilities} of agent $i \in [n]$ for the $j$th  bundle (for $j \in [n]$) in partition $P^k \in \mathcal{P}$. It is important to note that an agent with partition-based utilities can have different utilities for the exact same bundle, occurring in two distinct partitions.
We will denote a fair division instance by the tuple $\I=\langle[n], \mathcal{P}, \{u^k_{ij}\}_{i,j \in [n],k \in [m] }\rangle$.

For a given fair division instance, we define an \emph{allocation} to be an assignment of the $n$ bundles of a partition in $\mathcal{P}$ to the agents, such that every agent receives exactly one bundle. We assume that for any partition, the set of (admissible) allocations is specified by the $n!$ permutations of $n$ bundles among $n$ agents, and that the utility of an agent depends only on the partition and the bundle received and not to whom the remaining bundles are given. We refer to this property of a  fair division instance as the \emph{anonymity property}.
Therefore, we have a total of  $m \cdot n!$ many distinct admissible \emph{allocations} in a given fair division instance. Furthermore, in a given fair division instance, we define a \emph{lottery} to be a probability distribution over these allocations.

The overarching goal is to find a \emph{fair} and \emph{efficient} lottery among agents from the given set of admissible partitions. As mentioned, Cole and Tao \cite{cole2021existence} established the existence of fair and efficient lotteries for fair division instances with the anonymity property using Kakutani's fixed-point theorem~\cite{Kakutani1941}.
Since there are a total of $m \cdot n!$ many allocations, one can specify probabilities with which every allocation occurs in a lottery. This leads to a very convenient but also very inefficient way of representing a lottery via an exponential-dimensional  vector $(p_1, p_2, \dots, p_{m \cdot n!})$, where $p_i$ represents the probability with which the $i$-th allocation is chosen. This representation was used by Cole and Tao~\cite{cole2021existence} for their proof of existence, but it is clearly not suitable for studying the computational aspects of finding lotteries.
Instead, we will represent a lottery in the following  manner: Let $\p=\{p_k\}_{k \in [m]}$, where $p_k \in [0,1]$ denotes the probability with which  partition $P^k \in \mathcal{P}$ is selected in a lottery. The vector $\q=\{q^k_{ij}\}_{i,j \in [n], k \in [m]}$ of length $m \cdot n^2$ then specifies the full lottery, where $q^k_{ij}$ is the probability with which the lottery $\q$ assigns the $j$th bundle in partition $P^k$ to agent $i \in [n]$. The vectors $\p$ and $\q$ characterized by the following constraints.
\begin{align*}
    \sum_{i=1}^n q^k_{ij} &=p_k \ \text{for all} \ j \in [n], k \in [m] \tag{Every bundle is assigned to one agent}\\
     \sum_{j=1}^n q^k_{ij} &=p_k \ \text{for all} \ i \in [n], k \in [m] \tag{Every agent receives one bundle}\\
     \text{and,} \ \sum_{k=1}^m p_k &=1 
\end{align*}
We can now express the expected utility, $\e [u_i(\q)]$, for agent $i \in [n]$ in a lottery $\q$ as 
\[ \e [u_i(\q)] \coloneqq \sum_{k=1}^m \sum_{j=1}^n u^k_{ij} \cdot q^k_{ij} \enspace .\]
More generally, let \[u_i(\q;i') \coloneqq \sum_{k=1}^m \sum_{j=1}^n u^k_{ij} \cdot q^k_{i' j} \] denote the expected utility of agent $i$ for the bundle of agent $i'$  in the lottery $\q$. Observe that, we have $u_i(\q;i)=\e [u_i(\q)]$ for any agent $i \in [n]$.

Let us now define the standard notions of fairness and optimality in resource allocation settings. A lottery $\q$ is said to be \emph{ex-ante envy-free} if $u_i(\q;i) \geq u_i(\q;i')$ holds for all $i,i' \in [n]$. Furthermore, we say that $\q$ is  \emph{ex-ante Pareto-optimal} if there does not exist any other lottery $\widetilde{\q}$ such that $u_i(\widetilde{\q};i) \geq u_i(\q;i)$ holds for all $i \in [n]$, with a strict inequality for at least one agent $i \in [n]$. \emph{Social welfare} is a standard notion of measuring the collective welfare of an allocation. We define \emph{social welfare} of a lottery $\q$ as the sum of the expected utilities of all agents, i.e., $\mathrm{SW}(\q) = \sum_{i \in [n]} u_i(\q;i)$.

 \section{\PPAD-membership}

In this section, we show that the problem of finding an exact ex-ante envy-free and Pareto-optimal lottery in a given fair division instance belongs to the class \PPAD.
Our proof is based on (i) a significant simplification of the existence proof of Cole and Tao \cite{cole2021existence}, (ii) a characterization of \PPAD\ in terms of computing fixed points of piecewise linear arithmetic circuits due to Etessami and Yannakakis~\cite{EtessamiY2010-FIXP} (i.e.\ $\PPAD=\LinearFIXP$), and (iii) a framework for proving \FIXP\ and \PPAD-membership via convex optimization recently developed by Filos-Ratsikas~et~al.~\cite{Filos-RatsikasH2021-FIXP,Filos-RatsikasH2023-PPAD}. Formally, we obtain the following theorem.

\begin{theorem}
  The problem of finding an ex-ante envy-free and Pareto-optimal
  lottery in a fair division instance belongs to \PPAD.
  \label{THM:PPAD}
\end{theorem}

It is possible to adapt the existence proof of Cole and Tao (by changing to our succinct representation of lotteries) to obtain a proof of $\FIXP$ membership using the framework of Filos-Ratsikas~et~al~\cite{Filos-RatsikasH2021-FIXP}. The proof of Cole and Tao employs Kakutani's fixed point theorem to a correspondence defined on pairs consisting a lottery $\q$ and a vector of positive weights $\w \in W_\eps$ for the agents, from a closed set $W_\eps$. This correspondence maps $(\q,\w)$ to pairs $(\q',\w')$ such that $\q'$ is a lottery maximizing the weighted sum of utilities of the agents and where $\w'$ is obtained from $\w$ by translating each coordinate by a nonlinear function of the lottery $\q$ followed by a projection to the set $W_\eps$.

The maximization of the weighted sum of utilities may be phrased as a linear program and the projection may be phrased as a convex quadratic program. While both of these fall in the scope of the framework of Filos-Ratsikas~et~al.~\cite{Filos-RatsikasH2023-PPAD} for proving \PPAD-membership, the nonlinear transformation involved cannot be computed be a piecewise linear artihmetic circuit.

Our simplified proof involves only optimization of a linear program and the solution of a feasibility program with conditional linear constraints, together with operations computable by linear arithmetic circuits. In this case the framework Filos-Ratsikas~et~al.\ applies to give \PPAD-membership~\cite{Filos-RatsikasH2023-PPAD}.

Another framework for proving \PPAD-membership was also recently introduced by Papadimitriou, Vlatakis-Gkaragkounis and Zampetakis~\cite{PapadimitriouVZ2023-Kakutani}. With this framework, however, it would only be possible to directly prove \PPAD-membership for an \emph{approximate} version of the problem, rather than the exact problem.



In the remainder of this section, we let $\I=\langle[n], \mathcal{P}, \{u^k_{ij}\}_{i,j,k}\rangle$ denote a fair division instance with $n$~agents and $m=\abs{\mathcal{P}}$ partitions, where the utilities $u_{ij}^k$ are given as rational numbers.

\subsection{Fixed point formulation}
We first present our fixed point formulation for ex-ante envy-free and Pareto-optimal lotteries; afterwards we consider the implications for the computational complexity of the problem.

A standard technique for expressing the Pareto frontier of an
optimization problem, also employed by Cole and Tao, is the \emph{weighted sum
method}~\cite{Zadeh1963-Pareto}. Let $w_1,\dots,w_n>0$ be strictly positive weights. Then, any lottery $\q$ maximizing the weighted sum of utilities $\sum_{i=1}^n w_i u_i(\q;i)$ must be Pareto-optimal. Conversely, if $\q$ is a Pareto-optimal lottery, there are strictly positive weights such that $\q$ maximizes the weighted sum of utilities.

The task of maximizing the weighted sum of utilities can be expressed by the following linear program with decision variables $q_{ij}^k$ and $p_k$, and parameterized by the variables $w_i$.

\begin{equation}
  \begin{array}{ll@{}ll}
    \text{maximize}   & \sum\limits_{i=1}^n w_i \sum\limits_{k=1}^m \sum\limits_{j=1}^n u_{ij}^k q_{ij}^k&\\
    \text{subject to} & \sum\limits_{i=1}^n q^k_{ij} =p_k & \text{for all} \ j \in [n], k \in [m]  \\
    \displaystyle     & \sum\limits_{j=1}^n q^k_{ij} =p_k & \text{for all} \ i \in [n], k \in [m] \\
    \displaystyle     & \sum\limits_{k=1}^m p_k =1 &  \\   
    \displaystyle     & p_k\geq 0 & \text{for all}\ k \in [m] & \\
    \displaystyle     & q_{ij}^k\geq 0 & \text{for all} \ i,j \in [n], k \in [m] \\
  \end{array}
  \label{EQ:MaxWeightedSumLP}
\end{equation}

The proof of Cole and Tao~\cite{cole2021existence} shows the existence of positive weights such that any lottery $\q$ maximizing the corresponding weighted sum of utilities is also ex-ante envy-free.

We next define the following key quantity $0<\rho\leq \frac{1}{2}$ and state Lemma~\ref{lem:cole-tao} (proved in \cite{cole2021existence}) that will be used to place
restrictions of weights. For completeness we give the proof of the lemma in Appendix~\ref{app:413}, adapting the 
proof of \cite[Claim~4.13]{cole2021existence} to suit our representation of
lotteries.

\begin{definition}
  Let
  $J=\{(k,l,h,a,b) \in [m] \times [n]^4 \mid u_{la}^k < u_{lb}^k \text{ and } u_{ha}^k < u_{hb}^k\}$. We define $\rho$ as follows,
 \[   
\rho = 
     \begin{cases}
       \frac{1}{2} \min_{(k,l,h,a,b)\in J}  (u_{lb}^k-u_{la}^k)/(u_{hb}^k-u_{ha}^k) & \text{ if } J \neq \emptyset\\
         \frac{1}{2} & \text{ otherwise }
     \end{cases}
\]
\label{DEF:rho}
\end{definition}
\begin{lemma}[cf.\ {\cite[Claim~4.13]{cole2021existence}}] \label{lem:cole-tao}
  Suppose that $(\q,\p)$ is an optimal solution of
  LP~(\ref{EQ:MaxWeightedSumLP}). If $0<w_h \leq \rho w_l$ it follows
  that $u_l(\q;l) \geq u_l(\q;h)$ (i.e.\ that agent~$l$ does not envy
  agent~$h$).
\label{LEM:WeightInequalityImpliesEnvyFree}
\end{lemma}

Define $\eps=\rho^n/n$ and let
$W_\eps = \{w \in \RR^n \colon \sum_{i=1}^n w_i=1 \text{ and } w_i
\geq \eps\ \forall \ i \in [n]\}$.
We shall restrict the weights to belong to $W_\eps$, which in
particular, ensure that they are strictly positive. We consider the following feasibility problem with conditional linear constraints having decision variables
$w_i$, and parameterized by the variables $q_{ij}^k$. 
\begin{equation}
  \begin{array}{ll}
    \left[u_l(\q;h) - u_l(\q;l)>0\right]  \Rightarrow \left[w_h - \rho w_l \leq 0\right]  & \text{ for all } l,h \in [n]\\
    \sum\limits_{i=1}^n w_i = 1\\
    w_i \geq \eps & \text{ for all } i \in [n]
  \end{array}
  \label{EQ:WeightConditionalLP}
\end{equation}
Here, the conditional constraint $\left[u_l(\q;h) - u_l(\q;l)>0\right]  \Rightarrow \left[w_h - \rho w_l \leq 0\right]$ is satisfied if either $u_l(\q;h) - u_l(\q;l) \leq 0$ or $w_h - \rho w_l \leq 0$. In words, whenever agent~$l$ envies agent $h$ in the lottery $q$, a $\w$ solution of the system must satisfy $0< w_h \leq \rho w_l$, which is precisely the antecedent stated in Lemma~\ref{LEM:WeightInequalityImpliesEnvyFree}. We can think of the feasibility problem as a system of inequalities in variables $\w$, some of which may be ``disabled'' by inequalities expressed in the variables  $\q$.

In order to characterize the solvability of this feasibility program, it is convenient
to introduce the \emph{envy graph} of the lottery $\q$. Filos-Ratsikas~et~al.~\cite{Filos-RatsikasH2023-PPAD} consider feasibility programs as above in a general form and characterizes their solvability in terms of a \emph{feasibility graph}. In our case, this feasibility graph is exactly the same as the \emph{envy graph} defined next.
\begin{definition}[Envy graph]
  For a given lottery $\q$, denote by $\G(\q)$ the \emph{envy graph} with nodes
  $[n]$ and an arc $(l,h)$ whenever $u_l(\q;l) < u_l(\q;h)$, for all
  $l,h \in [n]$. We let $A(\G(\q))$ denote the set of arcs of $\G(\q)$.
\end{definition}
We can then precisely characterize the solvability of the feasibility problem~(\ref{EQ:WeightConditionalLP}) by the graph structure of~$\G(\q)$.
\begin{lemma}
  Suppose that $\q$ is a lottery such that $\G(\q)$ is acyclic. Then the feasibility program~(\ref{EQ:WeightConditionalLP}) is solvable.
\label{LEM:AcyclicImpliesSolvable}
\end{lemma}
\begin{proof}
  First note that the condition $u_l(\q;h) - u_l(\q;l)>0$ is satisfied
  precisely when $(l,h) \in A(\G(\q))$. Thus we are to find weights
  $w_i$ such that $w_h \leq \rho w_l$, whenever $(l,h) \in A(\G(\q))$.

  For $i \in [n]$, let $d_i$ denote the length of a longest path in
  $\G(\q)$ from node $i$ to a sink node, and define
  \[
    w_i = \frac{\rho^{d_i}}{\sum_{j=1}^n \rho^{d_j}} \text{ for all } i \in [n] \enspace .
  \]
  Clearly $\sum_{i=1}^n w_i = 1$, and since $d_i \leq n$ and
  $\rho \leq 1$ we also have $w_i \geq \eps$.  Suppose now that
  $(l,h) \in A(\G(\q))$. This means that $d_l \geq d_h+1$ and thus
  also $w_l \leq \rho w_h$. In conclusion, we have that the weights
   $w_i$ are a solution to the feasibility program~(\ref{EQ:WeightConditionalLP}).
\end{proof}
We can note that acyclicity of $\G(\q)$ is also necessary for the solvability of the feasibility program~(\ref{EQ:WeightConditionalLP}), since the
inequalities $w_l \leq \rho w_h$ given by the arcs $(l,h)$ of a cycle
in $\G(\q)$ are contradictory. But note also that if $\G(\q)$ contains a cycle, all agents in the cycle
will increase their utility if the lottery is shifted along the
cycle. We thus have the following simple but crucial observation.
\begin{observation}[cf.\ {\cite[Claim~4.8]{cole2021existence}}]
  If $\q$ is Pareto-optimal, the envy graph $\G(\q)$ is acyclic.
  \label{OBS:ParetoImpliesAcyclic}
\end{observation}
We can now conclude with the following fixed-point formulation, showing that a pair $(\q,\w)$ that is  simultaneously solving the linear program~(\ref{EQ:MaxWeightedSumLP}) and the feasibility problem~(\ref{EQ:WeightConditionalLP}) give an ex-ante envy-free and Pareto-optimal lottery.
\begin{proposition}
  Suppose that $\q$ is a lottery and $\w \in W_\eps$ are weights such
  that $\q$ is an optimal solution of the linear program
  LP~(\ref{EQ:MaxWeightedSumLP}) with respect to the weights $\w$, and
  $\w$ is a solution of the feasibility program of conditional linear
  constraints~(\ref{EQ:WeightConditionalLP}) with conditions given by
  $\q$ (note that the system is in fact solvable by the optimality of $\q$, Observation~\ref{OBS:ParetoImpliesAcyclic} and
  Lemma~\ref{LEM:AcyclicImpliesSolvable}). Then $\q$ is an ex-ante
envy-free and Pareto-optimal lottery.
\label{PROP:FixedPointFormulation}
\end{proposition}
\begin{proof}
  Since the weights $\w$ are strictly positive and $\q$ is an optimal
  solution of LP~(\ref{EQ:MaxWeightedSumLP}) it follows that $\q$ is
  Pareto-optimal. Suppose now for contradiction that there exists
  agents $l$ and $h$ such that agent $l$ envies agent $h$, that is,
  $u_l(\q;h) > u_l(\q;l)$. Since $\w$ is a solution to the
  system~(\ref{EQ:WeightConditionalLP}) with conditions given by $\q$ given it
  follows that $w_h \leq \rho w_l$. But then
  Lemma~\ref{LEM:WeightInequalityImpliesEnvyFree} gives
  $u_l(\q;h) \leq u_l(\q;l)$, contradicting the assumption. It thus follows that $\q$ must also be
  ex-ante envy-free.
\end{proof}

\subsection{\PPAD, \FIXP, and \LinearFIXP}
The complexity class \PPAD\ was originally defined in seminal work of Papadimitriou~\cite{Papadimitriou1994-TFNP} as the class of total \NP\ search problems reducible to a concrete problem called \textsc{End-Of-Line}. As mentioned above, to obtain result, we shall instead make use of a characterization of \PPAD\ in terms of computation of fixed points of functions computed by piecewise linear arithmetic circuits. Below we briefly introduce this characterization and refer to~\cite{EtessamiY2010-FIXP} for further details.

An arithmetic circuit is a circuit $C$ with gates computing binary operations belonging to the set $\{+,-,\ast,\div,\max,\min\}$ together with rational constants. The size of $C$ refers to the size of an encoding of $C$. A \emph{piecewise linear} arithmetic circuit $C$ restricts the allowable binary operations to the set $\{+,-,\max,\min\}$, but allows also for multiplication by rational constants.

The class \FIXP\ consists of (real-valued) search problems that reduce to finding a fixed point of a function $F \colon D \to D$, where $D$ is an explicitly given convex polytope and $F$ is a function computable by an algebraic circuit. By Brouwer's fixed point theorem such a fixed point is guaranteed to exist, thus making the search problem a total search problem. \LinearFIXP\ is the subclass obtained by restricting the arithmetic circuits to be piecewise linear. 

As defined above, the classes \FIXP\ and \LinearFIXP\ consist of real-valued search problems, which means that reductions must specify a real-valued function mapping fixed points of the function $F$ to solutions of the search problem. In the case when $F$ is computed by a piecewise linear arithmetic circuit $C$, there exists rational-valued fixed points of polynomial bitsize in the size of $C$~\cite[Theorem~5.2]{EtessamiY2010-FIXP}, which allows the use of ordinary polynomial-time reductions. With this convention, Etessami and Yannakakis~\cite{EtessamiY2010-FIXP} showed that $\PPAD = \LinearFIXP~\cite[Theorem~5.4]{EtessamiY2010-FIXP}$.

\subsection{\PPAD-membership via convex optimization}
From the characterization $\PPAD = \LinearFIXP$, in order to prove Theorem~\ref{THM:PPAD}, it is sufficient to reduce the task of computing an ex-ante envy-free and Pareto-optimal lottery to that of computing a fixed point of a piecewise linear arithmetic circuit defined on an explicitly given convex polytope. 

Constructing such a suitable circuit from scratch can potentially be a very challenging task, as many existing proofs of \PPAD-membership in the literature give evidence of. Recently however, Filos-Ratsikas~et~al.~\cite{Filos-RatsikasH2021-FIXP,Filos-RatsikasH2023-PPAD} introduced a general technique for proving \FIXP\ and \PPAD-membership, by which the arithmetic circuit defining the fixed point search problem can be augmented with \emph{pseudo-gates} that solve very general convex optimization problems. By a pseudo-gate is meant a (multi-input and multi-output) gate that is only required to compute the correct output at a fixed point of the full circuit. More precisely, the pseudo-gate is implemented by an arithmetic circuit using auxiliary variables, and when these auxiliary variables are in a fixed point, the pseudo-gate computes the correct output.
\begin{definition}[Pseudo-circuit]
    A pseudo-circuit with $n$ inputs and $m$ outputs is an arithmetic circuit $C$ computing a function $F \colon \RR^n \times [0,1]^\ell \to \RR^m \times [0,1]^\ell$. The output of $C$ on input $x \in \RR^n$ is any $y \in \RR^m$ such that there exists $z \in [0,1]^\ell$ such that $F(x,z)=(y,z)$. The variables $z \in [0,1]^\ell$ are called auxiliary variables.
\end{definition}
By a pseudo-gate is simply meant the use of a pseudo-circuit as a sub-circuit of larger pseudo-circuit, and where the auxiliary variables of the pseudo-gate is augmented to the auxiliary variables of the larger pseudo-circuit. The simple but crucial observation about pseudo-circuits is that, for the purpose of proving \FIXP\ and \PPAD-membership they are just as good as normal arithmetic circuits.

In the setting of proving \PPAD-membership, Filos-Ratsikas~et~al.~\cite{Filos-RatsikasH2023-PPAD} developed a pseudo-gate, coined the \emph{linear-OPT-gate}, implemented as a piecewise linear arithmetic circuit, that in particular can be used to solve both the linear program~(\ref{EQ:MaxWeightedSumLP}) and the feasibility problem~(\ref{EQ:WeightConditionalLP}). For the linear program~(\ref{EQ:MaxWeightedSumLP}) this is possible since the coefficients of all linear constraints are constants and that the coefficients of the objective function are linear functions of the parameter variables $\w$. For the feasibility program
 (\ref{EQ:WeightConditionalLP}) this is possible since the coefficients of all linear
constraints are constants and the that the antecedents of the conditional linear
constraints are given by a strict linear inequalities for functions computable by piecewise linear circuits applied to the parameter variables $\q$. We provide precise statements of the capabilities of the linear-OPT-gate in Appendix~\ref{app:linearOPT}.

\subsection{Proof of Theorem~\ref{THM:PPAD}}
We finally show how our fixed point formulation for ex-ante envy-free and Pareto-optimal lotteries in conjunction with the framework of Filos-Ratsikas~et~al.~\cite{Filos-RatsikasH2023-PPAD} allows for a simple proof of \PPAD\ membership for the problem of computing such lotteries.

The fixed point formulation of Proposition~\ref{PROP:FixedPointFormulation} amounts to finding
$(\q,\p,\w)$ such that $(\q,\p)$ is an optimal solution of the linear program~(\ref{EQ:MaxWeightedSumLP}), parametrized by $\w$, and such that $w$ is a solution to the feasibility program of conditional linear
constraints~(\ref{EQ:WeightConditionalLP}), parametrized by $\q$.

We thus build a piecewise linear arithmetic pseudo-circuit $C$ accomplishing both tasks. The circuit $C$ takes as input the variables $(\q,\p,\w)$. Using the linear-OPT-gate of~\cite{Filos-RatsikasH2023-PPAD} we let $C$ output $(\q',\p',\w')$ such that:
\begin{enumerate}
    \item $(\q',\p')$ is an optimal solution of LP~(\ref{EQ:MaxWeightedSumLP}) parametrized by $w$.
    \item If the feasibility program~(\ref{EQ:WeightConditionalLP}) parametrized by $q$ is feasible, then $w'$ is a solution.
\end{enumerate} 
Suppose now that $(\q,\p,\w)$ is a fixed point of the circuit $C$  (where also the auxiliary inputs of $C$ are assumed to be in a fixed point). Since $(\q,\p)$ is then an optimal solution of LP~\ref{EQ:MaxWeightedSumLP}) parametrized by $\w$, this means that $\q$ is Pareto-optimal by the weighted sum method. From Observation~\ref{OBS:ParetoImpliesAcyclic} and Lemma~\ref{LEM:AcyclicImpliesSolvable} we then have that the feasibility program~(\ref{EQ:WeightConditionalLP}) parametrized by $\q$ is in fact feasible, and this then means that $\w$ is a solution. By Proposition~\ref{PROP:FixedPointFormulation} we can then conclude that $\q$ is an ex-ante envy-free and Pareto-optimal lottery.

We have thus reduced the task of computing an ex-ante envy-free and Pareto-optimal lottery to the task of computing a fixed point of a piecewise linear arithmetic pseudo-circuit defined on a explicitly given convex polytope, thereby completing the proof.

\section{An efficient algorithm for constant number of agents}
In this section, we develop a very simple polynomial time algorithm for computing an ex-ante envy-free and Pareto-optimal lottery when the number of agents is constant. Consider a fair division instance $\mathcal{I}$ consisting of $n$ agents, a set of $m$ partitions $\mathcal{P}=\{P^1,P^2,..., P^m\}$, and agent utilities $u_{ij}^k$ for $i,j\in [n]$ and $k\in [m]$. The algorithm begins with evaluating the agents' valuations in the $n!$ possible allocations for each partition $P^k$ for $k\in [m]$. That is, we obtain $n!$ utility profiles in $\mathbb{R}^n$ for each partition and $m\cdot n!$ utility profiles overall. The Pareto-optimal lotteries are formed by faces of the convex hull of these utility profiles. 

Since the dimension $n$ is constant, the convex hull can be computed in polynomial time~\cite{Chazelle1993OptimalConvexHull}. We may then enumerate over the faces forming the Pareto-frontier. For each of these faces, we compute a hyperplane $H$ that contains the face. For such a hyperplane $H = \{x \in \mathbb{R}^n \mid w_1 x_1 + w_2 x_2 + \dots + w_n x_n = w_0\}$, we can determine whether it contains an envy-free lottery by linear programming.
\begin{equation*}
  \begin{array}{ll@{}ll}
    \text{Find}  & (\q,\p) &\\
    \text{subject to}& 
    \displaystyle   u_i(\q;i)\geq u_i(\q;i')\quad\quad &\text{for all} \ i,i' \in [n]\\
    \displaystyle & \sum\limits_{i=1}^n w_i u_i(\q;i) = w_0\\
    \displaystyle & \sum\limits_{i=1}^n q^k_{ij} =p_k &\text{for all} \ j \in [n], k \in [m]  \\
    \displaystyle  & \sum\limits_{j=1}^n q^k_{ij} =p_k & \text{for all} \ i \in [n], k \in [m] \\
    \displaystyle    & \sum\limits_{k=1}^m{p_k}=1 & \\
    \displaystyle     & p_k\geq 0 & \text{for all}\ k \in [m] & \\
    \displaystyle     & q_{ij}^k\geq 0 & \text{for all} \ i,j \in [n], k \in [m] \\
  \end{array}
\end{equation*}

Since we know that there does exist an ex-ante envy-free and Pareto-optimal lottery, at least one of these linear programs must be feasible. The next statement summarizes the discussion above.

\begin{theorem}
    For fair division instances with a constant number of agents, an ex-ante envy-free and Pareto-optimal allocation can be computed in polynomial time. 
\end{theorem}
\section{Ex-ante envy-free and Pareto-optimal lotteries of high social welfare}\label{sec:np-hardness}
As our last technical contribution, we study the problem of optimizing social welfare over ex-ante envy-free and Pareto-optimal allocation lotteries and prove the following statement for its decision version. 

\begin{theorem} \label{thm:NPhard}
    The problem of, given a fair division instance with partition-based utilities and $K>0$, deciding whether there exists an ex-ante envy-free and Pareto-optimal allocation lottery of social welfare at least $K$ is \NP-complete.
\end{theorem}

It is easy to see that the above problem belongs to the complexity class \NP. First, notice that it is trivial to check whether a given lottery $\q$ is ex-ante envy-free and has social welfare at least $K$. To verify Pareto optimality, it suffices to search for another lottery $\widetilde{\q}$ which gives to any agent expected utility at least as high as her expected utility in $\q$, maximizing the total excessive utility, through the following linear program:
\begin{equation*}
  \begin{array}{ll@{}ll}
    \text{maximize}  & \sum\limits_{i=1}^n{t_i} &\\
    \text{subject to}& \displaystyle u_i(\widetilde{\q};i)\geq u_i(\q;i)+t_i\quad\quad &\text{for all} \ i \in [n]\\
    \displaystyle & \sum\limits_{i=1}^n \widetilde{q}^k_{ij} =\widetilde{p}_k &\text{for all} \ j \in [n], k \in [m]  \\
    \displaystyle  & \sum\limits_{j=1}^n \widetilde{q}^k_{ij} =\widetilde{p}_k & \text{for all} \ i \in [n], k \in [m] \\
    \displaystyle    & \sum\limits_{k=1}^m{\widetilde{p}_k}=1 & \\
    \displaystyle     & \widetilde{q}_{ij}^k\geq 0 & \text{for all} \ i,j \in [n], k \in [m] \\
    \displaystyle    & \widetilde{p}_k\geq 0 & \text{for all} \ k\in [m]\\
    \displaystyle   & t_i\geq 0 & \text{for all} \ i\in [n]
  \end{array}
\end{equation*}
Clearly, the lottery $\widetilde{\q}$ Pareto-dominates $\q$ if and only if the objective value of the above linear program is strictly positive.

For proving \NP-hardness, we will develop a polynomial-time reduction from the classic \NP-complete problem \emph{Exact Cover by 3-Sets} (X3C) \cite{Karp1972} to our problem. X3C is defined as follows:
\begin{quote}
\textit{Instance}: A universe $\mathcal{E}= \{e_1, e_2, \dots, e_r\}$ of $r$ elements, a family $\mathcal{S}=\{S_1, S_2, \dots, S_t\}$ of triplets from $\mathcal{E}$, i.e., $S_j \subseteq \mathcal{E}$ with $\abs{S_j}=3$ for all $j \in [t]$. 

\textit{Question:} Does there exist an exact cover, i.e., a set of $r/3$ triplets from $\mathcal{S}$ that includes all elements of the universe $\mathcal{E}$? 
\end{quote}

For $i\in [r]$, we let $f_i$ denote the frequency of occurrence of element $e_i$, i.e., $f_i := \abs{\{j:e_i \in S_j\}}$. 

\subsection{The reduction} \label{sec:construction}
Starting with an instance $\phi$ of X3C, our reduction constructs a fair division instance $\In$ as follows. Instance $\In$ has the following set of $n=t+1+2t^2+3r$ agents.
\begin{itemize}
    \item $t+1$ \emph{base agents} $b_0, b_1, b_2, \dots, b_t$, 
    \item $2t$ \emph{set agents} $h_{j,1}, h_{j,2}, \dots, h_{j,2t}$, for every $j\in [t]$,
    \item three \emph{element agents} $v_i,w_i,$ and $z_i$ for every $i\in [r]$ 
\end{itemize}
 The set $\mathcal{P}$ of admissible partitions of an underlying set of items consists of $m=3t$ partitions $P_{j,c}$ for $j \in [t]$ and $c \in [3]$. We identify the $n$ bundles of partitions in accordance to the type of agents. So, each partition has $t+1$ bundles $B_0, B_1, \dots, B_t$, $2t$ bundles $H_{j,1}, \dots, H_{j,2t}$ for every $j \in [t]$, and three bundles $V_i, W_i$, and $Z_i$ for every $i \in [r]$.


The utilities of the agents for the bundles of partition $P_{j,c}$ for $j \in [t]$ and $c \in [3]$ are given in the following table. The table includes only non-zero utilities; any utility that is not specified in the table is equal to zero. In our reduction, we use parameters $\eps = \frac{1}{12t^2}, R=\frac{6t^3}{\eps}$, and $Q=\frac{6t}{\eps}$.

    
  \begin{center}
     \begin{tabular}{llll}
\toprule
  $c$ & agent & bundle & utility\\\midrule
 any   &  $b_0$ & $B_0$ & $R/t$ \\
 & $b_0$ & $B_j$ & $R$ \\
  & $b_j$ & $B_j$ & $R$ \\ 
& $z_i$ for $i \in [r]:e_i \in S_j$ & $Z_i$ & $1/f_i$ \\ \midrule
  $1$ & $h_{j,1}$ & $H_{j,1}$ & $Q$ \\
 & $v_i$ for $i \in [r]:e_i \in S_j$ & $V_i$ & $2$ \\
 & $z_i$ for $i \in [r]:e_i \in S_j$ & $V_i$ & $2/3$ \\ \midrule
 $2$ & $h_{j,2}$ & $H_{j,2}$ & $Q$ \\
  & $w_i$ for $i \in [r]:e_i \in S_j$ & $W_i$ & $2$ \\ 
   & $z_i$ for $i \in [r]:e_i \in S_j$ & $W_i$ & $(1+1/f_i)/f_i$  \\ \midrule
  $3$ & $h_{j,1}$ & $H_{j,1}$ & $Q(1-\eps)$ \\
 & $h_{j,2}$ & $H_{j,2}$ & $Q(1-\eps)$ \\
 & $h_{j,\ell}$ for $\ell = 3, ..., t+1$ & $H_{j,1}$ & $\eps$ \\
  & $h_{j,\ell}$ for $\ell = t+2, ..., 2t $ & $H_{j,2}$ & $\eps$ \\
  & $v_i$ for $i \in [r]:e_i \in S_j$ & $V_i$ & $2$ \\
   & $w_i$ for $i \in [r]:e_i \in S_j$ & $W_i$ & $2$ \\ 
    &  $z_i$ for $i \in [r]:e_i \in S_j$ & $V_i$ & $2/3$\\
    &  $z_i$ for $i \in [r]:e_i \in S_j$ & $W_i$ & $(1+1/f_i)/f_i$\\
 \bottomrule 
\end{tabular}
\end{center}

The reduction is clearly computable in polynomial time. We shall, without loss of generality, assume in the following that $t\geq 9$ and $r\leq3t$; otherwise, it is trivial to decide~$\phi$.

\begin{definition}[Canonical allocation]
For any partition $P_{j,c}$ with $j \in [t]$ and $c \in [3]$, we define the \emph{canonical} allocation as follows: bundle $B_k$ is assigned to base agent $b_k$ for $k \in \{0,1,\dots,t\}$, bundle $H_{j,\ell}$ is assigned to set agent $h_{j,\ell}$ for $\ell \in [2t]$, and, finally, bundle $V_i$ is assigned to element agent $v_i$, bundle $W_i$ is assigned to element agent $w_i$, and $Z_i$ is assigned to element agent $z_i$ for $i\in [r]$.
%
\end{definition}

\subsection{Proof of Theorem~\ref{thm:NPhard}}
We now prove the correctness of our reduction. We remark that when we refer to the expected social welfare achieved by a set $F$ of agents in a lottery $\q$, we refer to the sum of the expected utilities of agents in $F$ in $\q$. We begin by presenting two simple technical lemmas.



\begin{lemma} \label{lem:SW_set_element}
    Consider an ex-ante envy-free lottery of instance $\In$. For $j \in [t]$, the expected utility the set and element agents can get from each of the partitions $P_{j,1}$ or $P_{j,2}$, conditioned on the partition being the outcome of the lottery, is at most $Q+9$. Similarly, the expected utility the set and element agents can get from the partition $P_{j,3}$, conditioned on the partition being the outcome of the lottery, is at most $Q/2+9$.
\end{lemma}

\begin{proof}
    Consider a lottery $\q$ of instance $\In$ and let $j \in [t]$. In partitions $P_{j,1}$ and $P_{j,2}$, the set agents $h_{j,1}$ and $h_{j,2}$ can get a utility of at most $Q$, while the maximum utility from the element agents is $2$ from each of the six bundles $V_i$ or $W_i$ for each $e_i \in S_j$ and $1$ from each of the three bundles $Z_i$ for $e_i \in S_j$. Overall, the expected utility set and element agents get from each of the partitions $P_{j,1}$ and $P_{j,2}$, conditioned on the partition being the outcome of the lottery, is at most $Q+9$.

Now, assume that the lottery $\q$ is ex-ante envy-free. Consider partition $P_{j,3}$ and observe that the set agents $h_{j,1}, h_{j,3}, \dots, h_{j,t+1}$ have utility only for the bundle $H_{j,1}$ of partition $P_{j,3}$. Due to ex-ante envy-freeness, all these agents receive bundle $H_{j,1}$ with conditional probability $1/t$. Similarly, all set agents $h_{j,2}, h_{j,t+2}, \dots, h_{j,2t}$ receive bundle $H_{j,2}$ with conditional probability $1/t$. Also, the maximum utility that can be obtained in partition $P_{j,3}$ from the element agents is $2$ from each of the six bundles $V_i$ and $W_i$ for $e_i \in S_j$, and $1$ from each of the three bundles $Z_i$ for $e_i \in S_j$. Overall, using our assumption $t \geq 9$, which clearly also gives $Q \geq 36$ (recall that $Q=\frac{6t}{\eps}=72t^3$), we have that the expected utility set and element agents get from partition $P_{j,3}$, conditioned on the partition being the outcome of the lottery, is $\frac{2Q}{t} \cdot (1-\eps) + \frac{2t-2}{t} \cdot \eps +15 \leq \frac{Q}{3}+15 \leq \frac{Q}{2}+9$.
\end{proof}

\begin{lemma} \label{lem:base}
    In instance $\In$, for any partition $P_{j,c}$ with $j \in [t]$ and $c \in [3]$, any allocation in the support of a Pareto-optimal lottery, either assigns bundle $B_j$ to agent $b_0$ or assigns bundle $B_0$ to agent $b_0$ and bundle $B_j$ to agent $b_j$.
 \end{lemma}

\begin{proof}
    Consider a Pareto-optimal lottery $\q$ and assume, for the sake of contradiction, that it has in its support an allocation in partition $P_{j,c}$ for $j\in [t]$ and $c\in [3]$ which assigns to the base agent $b_0$ neither bundle $B_0$ nor bundle $B_j$. Then, since the base agent $b_0$ is the only one who can get positive utility from bundle $B_0$, the lottery $\widetilde{\q}$, which moves probability mass from the above allocation to the one in which the agent who gets bundle $B_0$ and the base agent $b_0$ have their bundles swapped, Pareto-dominates $\q$, contradicting its Pareto-optimality.

Now, assume that $\q$ has in its support an allocation in partition $P_{j,c}$ for $j\in [t]$ and $c\in [3]$, in which the base agent $b_0$ is assigned to bundle $B_0$ but bundle $B_j$ is not assigned to the base agent $b_j$. Then, since the base agent $b_j$ is the only agent besides $b_0$ who has positive utility for bundle $B_j$ at partition $P_{j,c}$, the lottery $\widetilde{\q}$, which moves probability mass from this allocation to the one in which the agent who gets bundle $B_j$ and the base agent $b_j$ have their bundles swapped, Pareto-dominates $\q$, again contradicting its Pareto-optimality. The lemma follows. 
\end{proof}

In the statements and proofs below, for a given lottery, we denote by $p_{j,c}$ the probability of partition $P_{j,c}$ being the outcome of the lottery, for $j\in [t]$ and $c\in [3]$.

The next lemma shows that ex-ante envy-free lotteries of high social welfare must place close to total probability~$\frac{1}{t}$ on the partitions $P_{j,1}$, $P_{j,2}$ and $P_{j,3}$, for each~$j\in [t]$.

\begin{lemma} \label{lem:uniform}
    In any ex-ante envy-free and Pareto-optimal lottery of instance $\In$ in which the base agents have social welfare at least $R+R/t+r/t-3$, it holds that $p_{j,1}+p_{j,2}+p_{j,3} \in [\frac{1-\eps}{t},\frac{1+\eps}{t}]$, for each $j \in [t]$.
\end{lemma}

\begin{proof}
    Consider an ex-ante envy-free and Pareto-optimal lottery $\q$ of social welfare at least $R+R/t+r/t-3$ for the base agents. For $j \in [t]$, denote by $\theta_j$ the total probability that bundle $B_0$ is assigned to agent $b_0$ in partitions $P_{j,1}, P_{j,2}$, and $P_{j,3}$. Also, let $\theta=\sum_{j \in [t]}\theta_j$ and $p_j=p_{j,1}+p_{j,2}+p_{j,3}$ for $j \in [t]$; clearly, $\sum_{j\in [t]}{p_j}=1$. By Pareto-optimality and Lemma~\ref{lem:base}, agent $b_0$ gets bundle $B_j$ in partitions $P_{j,1}, P_{j,2}$, and $P_{j,3}$ with total probability $p_j - \theta_j$ for each $j\in [t]$. Then, the expected utility of agent $b_0$ is 
    \[
    u_{b_0}(\q;b_0)=\theta \cdot R/t+ \sum_{j \in [t]}(p_j-\theta_j) \cdot R = \theta \cdot R/t + (1-\theta)\cdot R \enspace .
    \]

    Now, consider the base agent $b_j$ for $j \in [t]$. By Pareto-optimality and Lemma~\ref{lem:base}, this agent gets bundle $B_j$ with total probability $\theta_j$ in partitions $P_{j,1}, P_{j,2}$, and $P_{j,3}$ (i.e., whenever agent $b_0$ gets bundle $B_0$). Hence, we have $u_{b_0}(\q;b_j)=\theta_j \cdot R$. By ex-ante envy-freeness of the lottery $\q$, we have 
    \[
    \theta \cdot (R/t) + (1- \theta) \cdot R=u_{b_0}(\q;b_0) \geq u_{b_0}(\q;b_j) = \theta_j \cdot R \enspace ,
    \] i.e., 
    \begin{align} \label{eq:1}
        \theta_j \leq 1-\theta\cdot (1-1/t) 
    \end{align}
  for all $j \in [t]$.
   The expected utility of the base agent $b_j$ in partitions $P_{j,1}, P_{j,2}$, and $P_{j,3}$  is $\theta_j \cdot R$. In total, the social welfare of the base agents is
   \[
   u_{b_0}(\q;b_0)+ \sum_{j \in [t]}u_{b_j}(\q;b_j) = \theta \cdot R/t+  (1-\theta) \cdot R + \sum_{j \in [t]}\theta_j \cdot R = R+ \theta \cdot R/t \enspace .
   \]
   Since the social welfare of the base agents is at least $R+R/t+r/t-3$, we have 
   \begin{align} \label{eq:2}
       \theta \geq 1 - \frac{3t-r}{R}.
   \end{align}
   We now claim that 
   \begin{align} \label{eq:3}
       p_j \leq \theta_j+ \frac{3t-r}{R}
   \end{align}
    for every $j \in [t]$. Indeed, assume that this is not the case and, instead, $\theta_{j^*} < p_{j^*} - \frac{3t-r}{R}$ for some $j^* \in [t]$. Using the inequality $\theta_j \leq p_j$ for $j \in [t] \setminus \{j^*\}$, and summing these inequalities up, we get $\theta = \sum_{j \in [t]} \theta_j< \sum_{j \in [t]}p_j - \frac{3t-r}{R}= 1- \frac{3t-r}{R}$, contradicting inequality~(\ref{eq:2}). Using equations~(\ref{eq:3}), (\ref{eq:1}), and (\ref{eq:2}) (in this order), we have 
    \begin{align*}
        p_j &\leq \theta_j+\frac{3t-r}{R} \leq 1-\theta \cdot (1-1/t)+ \frac{3t-r}{R}  \leq 1- \left(1-\frac{3t-r}{R}\right) \cdot \left(1-1/t\right) + \frac{3t-r}{R} \\
        &= \frac{1}{t} + \frac{3t-r}{R} \cdot \left(2-\frac{1}{t}\right)  \leq \frac{1}{t} + \frac{\eps}{t^2},
        \end{align*}
        which clearly implies the desired upper bound on $p_j$. The last inequality follows by the definition of $R$ (recall that $R=\frac{6t^3}{\eps}$). Then, 
        \begin{align*}
            p_j = 1-\sum_{j'\in [t]\setminus\{j\}}p_{j'} \geq  1-(t-1)\cdot \left(\frac{1}{t}+\frac{\eps}{t^2}\right) \geq \frac{1-\eps}{t},
        \end{align*}
        which completes the proof.
 \end{proof}

Our next technical lemma shows that in Pareto-optimal lotteries, almost all of the total probability given to the two partitions~$P_{j,1}$ and $P_{j,2}$ is given to one of them.

\begin{lemma} \label{lem:PO}
    Any Pareto-optimal lottery in instance $\In$ satisfies $\max\{p_{j,1},p_{j,2}\} \geq (1-~\eps)(p_{j,1}+p_{j,2})$, for all $j \in [t]$.
\end{lemma}

\begin{proof}
The claim is clear for $j \in [t]$ such that $p_{j,1}=0$ or $p_{j,2}=0$. So, consider a Pareto-optimal lottery $\q$ with $p_{j^*,1}>0$ and $ p_{j^*,2}>0$ for some $j^* \in [t]$. For the sake of contradiction, let $\max\{p_{j^*,1},p_{j^*,2}\} < (1-\eps)(p_{j^*,1}+p_{j^*,2})$.

We construct the lottery $\widetilde{\q}$ which has the same probability as $\q$ for every allocation in partition $P_{j,c}$ with $j \in [t] \setminus \{j^*\}$ and $c \in [3]$, probability $0$ for every allocation in partition $P_{j^*,1}$ and $P_{j^*,2}$, and a probability for the canonical allocation of partition $P_{j^*,3}$ that is $p_{j^*,1}+p_{j^*,2}$ higher than the corresponding probability in $\q$. In other words, compared to $\q$, $\widetilde{\q}$ has moved the probability mass of allocations of partitions $P_{j^*,1}$ and $P_{j^*,2}$ to the canonical allocation of partition $P_{j^*,3}$. Note that, the base agents, the set agents different than $h_{j^*,1}$ and $h_{j^*,2}$, and the element agents have at least as high expected utility in $\widetilde{\q}$ as in $\q$. Agent $h_{j^*,1}$ (respectively, $h_{j^*,2}$) gets utility (at most) $Q$ from partition $P_{j^*,1}$ (respectively, $P_{j^*,2}$), and utility $Q(1-\eps)$ from (the canonical allocation of) partition $P_{j^*,3}$. Hence, the increase of expected utility from $\q$ to $\widetilde{\q}$ for agents $h_{j^*,1}$ and $h_{j^*,2}$ is $Q(1-\eps)(p_{j^*,1}+p_{j^*,2}) - Q \cdot p_{j^*,1}$, and $Q(1-\eps)(p_{j^*,1}+p_{j^*,2}) - Q \cdot p_{j^*,2}$, respectively. By our assumption on $p_{j^*,1}$ and $p_{j^*,2}$, both quantities are strictly positive, contradicting the Pareto-optimality of $\q$.    
\end{proof}

Together, the lemmas above allow us to show that any ex-ante envy-free and Pareto-optimal lottery of high social welfare has an {\em almost combinatorial structure}. We remark that this is the crucial property of our reduction that essentially allows us to embed the combinatorial search space of X3C into the continuous space of allocation lotteries. Namely, for each $j \in [t]$, the lottery must give a probability mass of almost $1/t$ to one of the partitions $P_{j,1}$ or $P_{j,2}$, and a probability mass of almost $0$ to the other. This is stated more precisely in the following lemma.

\begin{lemma} \label{lem:binary_switch}
In instance $\In$, any ex-ante envy-free and Pareto-optimal lottery, in which the expected social welfare of the set and element agents is at least $Q+6+r/t$ and the expected social welfare of the base agents is at least $R+R/t+r/t-3$, satisfies $\max\{p_{j,1},p_{j,2}\} \geq \frac{1-3\eps}{t}$, $\min\{p_{j,1},p_{j,2}\} \leq \frac{2\eps}{t}$, and $p_{j,3}\leq \frac{\eps}{t}$ for each $j \in [t]$.
\end{lemma}

\begin{proof}
Let $\q$ be an  ex-ante envy-free and Pareto-optimal lottery in instance $\In$ with the stated social welfare guarantees. Let $j \in [t]$. By Pareto optimality and the required bound on the social welfare of the base agents, Lemmas~\ref{lem:uniform} and  \ref{lem:PO} yield $\min\{p_{j,1},p_{j,2}\} \leq \eps \cdot (p_{j,1}+p_{j,2}) \leq \eps \cdot \frac{1+\eps}{t} \leq \frac{2 \eps}{t}$. Next we prove the bounds on $p_{j,3}$ and $\max \{p_{j,1}, p_{j,2}\}$.

By the bounds on the expected utility of the set and element agents in partitions $P_{j,1}$, $P_{j,2}$, and $P_{j,3}$ from Lemma~\ref{lem:SW_set_element}, the probabilities of allocations in these partitions in the support of lottery $\q$, and 
    the bound on the social welfare of these agents in the statement of the lemma, we have 
   \[
   Q+6+\frac{r}{t} \leq (Q+9) \cdot \left(1- \sum_{j \in [t]}p_{j,3}\right)+\left(\frac{Q}{2}+9\right) \cdot \sum_{j \in [t]} p_{j,3}
   \]
   and, thus,
   \[
   \sum_{j \in [t]} p_{j,3} \leq \frac{2}{Q} \cdot \left(3-\frac{r}{t}\right) \leq \frac{\eps}{t} \enspace .
   \]
   Trivially, this implies the desired bound $p_{j,3} \leq \frac{\eps}{t}$ for $j \in [t]$. The last inequality follows by the definition of $Q$ (recall that $Q=\frac{6t}{\eps}$). Now, by Lemma~\ref{lem:uniform}, we get $p_{j,1}+p_{j,2} \geq \frac{1-2\eps}{t}$ and, by Lemma~\ref{lem:PO}, we have $\max\{p_{j,1},p_{j,2}\} \geq (1-\eps) \cdot \frac{1-2 \eps}{t} \geq \frac{1-3\eps}{t}$ for $j \in [t]$.
\end{proof}

Our arguments in the next two lemmas use a particular type of non-canonical allocations.

\begin{definition}
An allocation in partition $P_{j,c}$ for $j\in [t]$ and $c\in [3]$ is called {\em defective} if there is $i\in [r]$ such that $e_i\in S_j$ and agent $z_i$ is not assigned bundle $Z_i$.
\end{definition}
In the proofs of the next two lemmas, for a given lottery, we will denote by $\gamma_{j,c}$ the probability mass put on defective allocations in partition $P_{j,c}$ for $j\in [t]$ and $c\in [3]$. We denote by $\gamma$ the total probability mass put on defective allocations, i.e., $\gamma=\sum_{j\in [t]}{\left(\gamma_{j,1}+\gamma_{j,2}+\gamma_{j,3}\right)}$.

Our next technical lemma proves an upper bound on the probability mass put by any Pareto-optimal lottery with high enough social welfare on defective allocations.

\begin{lemma} \label{lem:noncanonical}
    In instance $\In$, any ex-ante envy-free and Pareto-optimal lottery with social welfare at least $R+R/t+r/t-3$ for the base agents and at least $Q+6+r/t$ for the set and element agents, must put a probability mass of at most $5 \eps$ on defective allocations.
\end{lemma}
\begin{proof}
Let $\q$ be an ex-ante envy-free and Pareto-optimal lottery in instance $\In$ with the stated social welfare guarantees. Let $j\in [t]$ and consider an allocation in partition $P_{j,1}$ in the support of lottery $\q$. We claim that for every $i\in [r]$ such that $e_i\in S_j$, this allocation either (1) assigns bundle $Z_i$ to the element agent $z_i$ and the bundle $V_i$ to the element agent $v_i$ or (2) assigns bundle $V_i$ to the element agent $z_i$. Indeed, for the sake of contradiction, assume that for some $i^*\in [r]$ such that $e_{i^*}\in S_j$, the element agent $z_{i^*}$ is assigned neither bundle $Z_{i^*}$ nor bundle $V_{i^*}$. Then, the lottery $\widetilde\q$, which moves probability mass from this allocation to the allocation in which the agent who gets bundle $Z_{i^*}$ and agent $z_{i^*}$ have their bundles swapped, Pareto-dominates lottery $\q$ (notice that bundle $Z_{i^*}$ gives zero utility to any other agent besides $z_{i^*}$, and agent $z_{i^*}$ gets zero utility from any bundle different than $Z_{i^*}$ and $V_{i^*}$), contradicting its Pareto-optimality. Now, again for the sake of contradiction, assume that agent $z_{i^*}$ is assigned bundle $Z_{i^*}$ but agent $v_{i^*}$ does not get bundle $V_{i^*}$. Then, the lottery $\widetilde\q$, which moves probability mass from this allocation to the allocation in which the agent who gets bundle $V_{i^*}$ and agent $v_{i^*}$ have their bundles swapped, Pareto-dominates lottery $\q$ (notice that bundle $V_{i^*}$ gives zero utility to any other agent besides agents $v_{i^*}$ and $z_{i^*}$, and agent $v_{i^*}$ gets zero utility from any bundle different than $V_{i^*}$), contradicting its Pareto-optimality. 

Thus, a non-defective allocation in partition $P_{j,1}$ in the support of lottery $\q$ gives utility $\sum_{i\in [r]:e_i\in S_j}{(2+1/f_i)}$ to the element agents. In a defective allocation in partition $P_{j,1}$ in the support of lottery $q$, the element agent $z_{i^*}$ has utility at most $(1+1/f_{i^*})/f_{i^*}$ instead of $1/f_{i^*}$ and the element agent $v_{i^*}$ has utility $0$ instead of $2$, for some $i^*\in [r]$ such that $e_{i^*}\in S_j$. Thus, a defective allocation in partition $P_{j,1}$ in the support of lottery $\q$ gives utility at most $\sum_{i\in [r]:e_i\in S_j}{(2+1/f_i)}-1$ to the element agents. 

Following analogous reasoning to the two paragraphs above, we can show that a non-defective allocation in partition $P_{j,2}$ (respectively, $P_{j,3}$) in the support of lottery $\q$ gives utility $\sum_{i\in [r]:e_i\in S_j}{(2+1/f_i)}$ (respectively, $\sum_{i\in [r]:e_i\in S_j}{(4+1/f_i)}$) to the element agents, and a defective allocation in partition $P_{j,2}$ (respectively, $P_{j,3}$) in the support of lottery $\q$ gives utility $\sum_{i\in [r]:e_i\in S_j}{(2+1/f_i)}-1$ (respectively, $\sum_{i\in [r]:e_i\in S_j}{(4+1/f_i)}-1$) to the element agents. 

We are now ready to upper-bound the social welfare of the element agents in lottery $\q$ by
\begin{align}\nonumber
    &\sum_{j \in [t]}{\left((p_{j,1}-\gamma_{j,1})\cdot \sum_{i\in [r]: e_i \in S_j}{\left(2+\frac{1}{f_i}\right)}+(p_{j,2}-\gamma_{j,2})\cdot \sum_{i\in [r]: e_i \in S_j}{\left(2+\frac{1}{f_i}\right)}\right.}\\\nonumber
    &\quad\quad {+(p_{j,3}-\gamma_{j,3})\cdot \sum_{i\in [r]: e_i \in S_j}{\left(4+\frac{1}{f_i}\right)}+\gamma_{j,1}\cdot \left(\sum_{i\in [r]: e_i \in S_j}{\left(2+\frac{1}{f_i}\right)}-1\right)}\\\nonumber
    &\quad\quad {\left.+\gamma_{j,2}\cdot \left(\sum_{i\in [r]: e_i \in S_j}{\left(2+\frac{1}{f_i}\right)}-1\right)+\gamma_{j,3}\cdot \left(\sum_{i\in [r]: e_i \in S_j}{\left(4+\frac{1}{f_i}\right)}-1\right)\right)}\\\nonumber
    & = \sum_{j \in [t]}{(6p_{j,1}+6p_{j,2}+12p_{j,3})}+ \sum_{j\in [t]}{\sum_{i\in [r]: e_i \in S_j}{ (p_{j,1}+p_{j,2}+p_{j,3})\cdot \frac{1}{f_i}}}-3\cdot \sum_{j \in [t]}{(\gamma_{j,1}+ \gamma_{j,2}+\gamma_{j,3})}\\\label{eq:sw-from-element-agents}
    & \leq 6+12\eps + \frac{1+\eps}{t}\cdot \sum_{i \in [r]}{\sum_{j\in [t]: e_i \in S_j}{\frac{1}{f_i}}} -3\gamma  = 6 + 12 \eps + (1+\eps) \cdot \frac{r}{t} - 3\gamma\leq 6+\frac{r}{t}+15\eps-3\gamma.
\end{align}
The first inequality follows using the inequality $p_{j,1}+p_{j,2}+p_{j,3}\leq \frac{1+\eps}{t}$ from Lemma~\ref{lem:uniform} and $p_{j,e}\leq \frac{\eps}{t}$ from Lemma~\ref{lem:binary_switch}. The second equality follows by the definition of $f_i$, and the last inequality follows by our assumption that $r\leq 3t$.

Now, recall (by our construction) that the social welfare of the set agents is no more than $Q$ and, hence, the lottery $\q$ has a social welfare of at least $6+r/t$ from the element agents. Using this observation and the upper bound on the social welfare of the element agents in  (\ref{eq:sw-from-element-agents}), we obtain that $\gamma \leq  5 \eps$, as desired.
\end{proof}

We are now ready to prove the soundness and completeness of our reduction. This is done in Lemmas~\ref{lem:soundness} and \ref{lem:completeness}, respectively, which complete the proof of Theorem~\ref{thm:NPhard}.
\begin{lemma} \label{lem:soundness}
    If instance $\In$ admits an ex-ante envy-free and Pareto-optimal lottery of social welfare at least $R+R/t+Q+6+r/t$, then instance $\phi$ has an exact cover.
\end{lemma}

\begin{proof}
Let $\q$ be an ex-ante envy-free and Pareto-optimal lottery in instance $\In$ with the stated social welfare guarantee. By Lemma~\ref{lem:SW_set_element}, the expected utility set and element agents have is at most $Q+9$. Hence, the social welfare of the base agents is at least $R+R/t+r/t-3$ and the conditions of Lemma~\ref{lem:uniform} are satisfied. Now, observe that for $i\in [r]$, the utility of agent $z_i$ is at most $1/f_i$ in any non-defective allocation in partitions $P_{j,c}$ for $j\in [t]$ such that $e_i\in S_j$ and $c\in [3]$, while it is at most $\max\left\{1/f_i,2/3,(1+1/f_i)/f_i\right\}\leq 1+1/f_i$ in any defective allocation in partition $P_{j,c}$ for $j\in [t]$ such that $e_i\in S_j$ and $c\in [3]$. Clearly, the utility of agent $z_i$ is zero in any allocation in partition $P_{j,c}$ for $j\in [t]$ such that $e_i\not\in S_j$. Thus, the expected utility of agent $z_i$ for $i\in [r]$ is 
   \begin{align}\nonumber
       u_{z_i}(\q;z_i) & \leq  \sum_{j \in [t]:e_i \in S_j}\left((p_{j,1}+p_{j,2}+p_{j,3}-\gamma_{j,1}-\gamma_{j,2}-\gamma_{j,3}) \cdot \frac{1}{f_i} + (\gamma_{j,1}+\gamma_{j,2}+\gamma_{j,3}) \cdot \left(1+\frac{1}{f_i}\right)\right) \\\nonumber
       &= \sum_{j\in [t]:e_i\in S_j}{(p_{j,1}+p_{j,2}+p_{j,3})\cdot \frac{1}{f_i}}+\sum_{j\in [t]:e_i\in S_j}{(\gamma_{j,1}+\gamma_{j,2}+\gamma_{j,3})}\\\label{eq:z_i-exp-utility-ub}
       &\leq \frac{1+\eps}{t}  \sum_{j \in [t]:e_{i} \in S_j}{\frac{1}{f_i}} +  \gamma \leq \frac{1}{t} + \frac{\eps}{t}+5\eps < \frac{1}{t}+\frac{1}{2t^2}.
   \end{align}
The second inequality follows by Lemma~\ref{lem:uniform} which asserts that $p_{j,1}+p_{j,2}+p_{j,3}\leq \frac{1+\eps}{t}$, the third one by the definition of $f_i$ and Lemma~\ref{lem:noncanonical}, and the last one by the definition of $\eps$ (recall that $\eps=\frac{1}{12t^2}$) and since $t\geq 9$.

On the other hand, notice that the bundle $B_0$ gives utility $R/t$ only to the base agent $b_0$, while for $j\in [t]$ and $c\in [3]$, the only bundle among $B_1, B_2, ..., B_t$ that gives non-zero utility (equal to $R$) to some base agent is bundle $B_j$. Thus, the social welfare of the base agents is at most $R+R/t$ and, hence, the social welfare of the set and element agents in lottery $\q$ is at least $Q+6+r/t$. Together with the properties of ex-ante envy-freeness and Pareto-optimality and the lower bound on the social welfare of the base agents claimed above, the conditions of Lemma~\ref{lem:binary_switch} are satisfied, meaning that the lottery $\q$ has the combinatorial structure indicated by it. 

Define $C=\{j\in [t]: p_{j,1} \geq \frac{1-3\eps}{t}\}$. We will show that $C$ forms an exact cover of $\phi$. For the sake of contradiction, assume otherwise that there exists an element $e_{i^*}$ for some $i^* \in [r]$ that is included in either none or in at least two sets $S_j$ such that $j \in C$. We distinguish between two cases:
    
\paragraph{Case 1.} If $e_{i^*}$ is not included in any set $S_j$ such that $j \in C$, then $p_{j,1} < \frac{1-3\eps}{t}$ and, by Lemma~\ref{lem:binary_switch}, $p_{j,2} \geq \frac{1-3\eps}{t}$ for all $j \in [t]$ such that $e_{i^*} \in S_j$. Now, notice that agent $w_{i^*}$ is assigned bundle $W_{i^*}$ (for which agent $z_{i^*}$ has utility $(1+1/f_{i^*})/f_{i^*}$) in every non-defective allocation in partition $P_{j,2}$ for $j\in [t]$ such that $e_{i^*}\in S_j$. Thus, the expected utility agent $z_{i^*}$ has for the bundle assigned to agent $w_{i^*}$ is 
\begin{align}\nonumber
         u_{z_{i^*}}(\q;w_{i^*}) & \geq \sum_{j \in [t]:e_{i^*} \in S_j}{(p_{j,2} - \gamma_{j,2})\cdot \left(1+\frac{1}{f_{i^*}}\right) \cdot \frac{1}{f_{i^*}}} \geq \left(1+\frac{1}{t}\right)  \cdot \sum_{j \in [t]:e_{i^*} \in S_j}{(p_{j,2} - \gamma_{j,2}) \cdot \frac{1}{f_{i^*}}} \\\nonumber  
         &\geq  \left(1+\frac{1}{t}\right)  \cdot \frac{1-3\eps}{t} \cdot\sum_{j \in [t]:e_{i^*} \in S_j}{ \frac{1}{f_{i^*}}} - \left(1+\frac{1}{t}\right) \cdot \sum_{j \in [t]:e_{i^*}\in S_j} \frac{\gamma_{j,2}}{f_{i^*}}\\\label{eq:z_i-exp-utility-lb}
         & \geq \left(1+\frac{1}{t}\right) \cdot \frac{1-3\eps}{t}-\left(1+\frac{1}{t}\right) \cdot \gamma \geq \frac{1}{t}+\frac{1}{t^2} -\frac{3\eps}{t^2}-\frac{8\eps}{t}-5\eps > \frac{1}{t}+\frac{1}{2t^2}.
    \end{align}
The second inequality follows since $f_{i^*}\leq t$ by definition, the third one since $p_{j,2}\geq \frac{1-3\eps}
{t}$, the fourth one by the definitions of $\gamma$ and $f_{i^*}$, the fifth one by Lemma~\ref{lem:noncanonical}, and the last one by the definition of $\eps$ (recall that $\eps=\frac{1}{12t^2}$ and since $t\geq 9$. By inequalities (\ref{eq:z_i-exp-utility-ub}) and (\ref{eq:z_i-exp-utility-lb}), we obtain that $u_{z_{i^*}}(\q;z_{i^*})<u_{z_{i^*}}(\q;w_{i^*})$, meaning that agent $z_{i^*}$ is envious of agent $w_{i^*}$, a contradiction.

\paragraph{Case 2.} Let $D=\{j\in C: e_{i^*}\in S_j\}$ and assume that $|D|\geq 2$. Since $D\subseteq C$, we have $p_{j,1}\geq \frac{1-3\eps}{t}$ for every $j \in D$. Notice that agent $v_{i^*}$ is assigned bundle $V_{i^*}$ (for which agent $z_{i^*}$ has utility $2/3$) in every non-defective allocation in partition $P_{j,1}$ for $j\in [t]$ such that $e_{i^*}\in S_j$. Thus, the expected utility agent $z_{i^*}$ has for the bundle assigned to agent $v_{i^*}$ is
   \begin{align}\label{eq:z_i-exp-utility-lb-2}
         u_{z_{i^*}}(\q;v_{i^*}) \geq  \sum_{j\in D}  \frac{2}{3}  \cdot (p_{j,1} - \gamma_{j,1}) \geq \frac{4}{3} \cdot \frac{1-3\eps}{t} - \frac{2}{3} \cdot  \gamma 
         \geq \frac{4}{3t}-\frac{4\eps}{t}-\frac{10}{3}\eps> \frac{1}{t} +\frac{1}{2t^2}.
    \end{align}
The third inequality follows by Lemma~\ref{lem:noncanonical} and the last one by the definition of $\eps$ (recall that $\eps=\frac{1}{12t^2}$) and since $t\geq 9$. By inequalities (\ref{eq:z_i-exp-utility-ub}) and (\ref{eq:z_i-exp-utility-lb-2}), we obtain that $u_{z_{i^*}}(\q;z_{i^*})<u_{z_{i^*}}(\q;w_{i^*})$, again meaning that agent $z_{i^*}$ is envious of agent $w_{i^*}$, a contradiction.
\end{proof}

\begin{lemma} \label{lem:completeness}
    If instance $\phi$ has an exact cover, then instance $\In$ admits an ex-ante envy-free and Pareto-optimal lottery of social welfare at least $R+R/t+Q+6+r/t$.
\end{lemma}
\begin{proof}
Let $C$ be an exact cover of instance $\phi$. Construct the lottery for instance $\In$ which has $p_{j,1}=1/t$ and $p_{j,2}=0$ for $j \in C$, $p_{j,1}=0$ and $p_{j,2}=1/t$ for $j \notin C$, and $p_{j,3}=0$ for $j \in [t]$, and uses the canonical allocations only. 
    
We first justify the claimed social welfare of the lottery. Clearly, the base agent $b_0$ has expected utility $R/t$. For $j\in [t]$, the base agent $b_j$ gets utility $R$ from partitions $P_{j,1}$ and $P_{j,2}$ only. Each of them has probability $1/t$. So, the overall expected utility of the base agents $b_1, ..., b_j$ is $R$. For $j\in C$, agent $h_{j,1}$ is the only set agent who gets utility $Q$ from partition $P_{j,1}$. For $j\not\in C$, agent $h_{j,2}$ is the only set agent who gets utility $Q$ from partition $P_{j,2}$. So, the overall expected utility of set agents is $Q$. For $i\in [r]$, the element agent $z_i$ gets utility $1/f_i$ with probability $1/t$ from either partition $P_{j,1}$ or partition $P_{j,2}$ for every $j\in [t]$ such that $e_i\in S_j$. Thus, agent $z_i$ has expected utility $1/t$; so, the $r$ element agents $z_i$ for $i\in [r]$ have overall expected utility $r/t$. For $j\in C$, partition $P_{j,1}$ gives a utility of $6$ to the three element agents $v_i$ corresponding to each element $e_i\in S_j$. Similarly, for $j\not\in C$, partition $P_{j,2}$ gives a utility of $6$ to the three agents $w_i$ corresponding to each element $e_i\in S_j$. So, the total expected utility of element agents $v_i$ and $w_i$ is $6$.

Let us now examine ex-ante envy-freeness. Notice that in any allocation in the support of the lottery, all base agents besides agent $b_0$, all set agents, and all element agents besides agent $z_i$ for $i\in [r]$ are assigned the bundle that gives them maximum utility. So, to justify ex-ante envy-freeness, we just need to examine whether agent $b_0$ envies agent $b_j$ for $j\in [t]$ and whether agent $z_i$ envies agents $v_i$ and $w_i$ for $i\in [r]$. Agent $b_0$ gets utility $R/t$ from bundle $B_0$ in any allocation of the lottery. Her utility for the bundle $B_j$ assigned to agent $b_j$ in partitions $P_{j,1}$ and $P_{j,2}$ is $R$ and is thus non-envious as one of them is part of the lottery with probability $1/t$. For $i\in [r]$, agent $z_i$ gets utility 
$1/f_i$ in each allocation of the lottery corresponding to a partition $P_{j,1}$ or $P_{j,2}$ with $e_i\in S_j$. Notice that there are exactly $f_i$ such partitions appearing in the lottery with probability $1/t$, for an expected utility of $1/t$ for agent $z_i$. Now, notice that agent $v_i$ gets bundle $V_i$ of partition $P_{j,1}$ for the single $j\in C$ such that $e_i\in S_j$. This happens with probability $1/t$, and agent $z_i$ has a value of only $2/3$ for this bundle. For $i\in [r]$, agent $w_i$ gets bundle $W_i$ of partition $P_{j,2}$ for $f_i-1$ different values of $j\not\in C$ so that $e_i\in S_j$, i.e., with probability $\frac{f_i-1}{t}$. Agent $z_i$ has total utility $\frac{f_i-1}{t}\cdot (1+1/f_i)/f_i < 1/t$ for the bundles assigned to agent $w_i$. Thus, indeed, for each $i\in [r]$, agent $z_i$ does not envy agents $v_i$ and $w_i$.

To prove Pareto optimality, notice that the base agents get their maximum expected utility of $R+R/t$, the set agents get their maximum expected utility of $Q$, and the element agents $v_i$ and $w_i$ for $i\in [r]$ get their maximum expected utility of $6$. Then, any allocation in which the utility of some agent $z_i$ increases should harm some other element agent. This concludes the proof of the lemma.
\end{proof}

\section{Conclusion}

In this work, we considered the general setting of the problem of dividing indivisible items in a fair and efficient manner to agents having partition-based utilities. We have shown membership of the total problem of finding ex-ante envy-free and Pareto-optimal allocation lotteries in the class \PPAD. We consider settling the precise computational complexity of the problem an important question. From an algorithmic perspective it would also be very interesting to see if Lemke's algorithm~\cite{Lemke1964-Bimatrix} could be adapted to solve the problem, as this would likely lead to a practical algorithm.



\bibliography{ref}
\appendix
\section{Proof of Lemma~\ref{LEM:WeightInequalityImpliesEnvyFree}} \label{app:413}
\begin{proof}
  By the Birkhoff–von~Neumann theorem, for any $k$, there exists a
  probability distribution $\{\alpha^k_\pi\}_{\pi \in S_n}$ over
  permutations on $[n]$ such that
  \begin{equation}
    q_{ij}^k = p_k \sum_{\substack{\pi \in S_n \\ \pi(i)=j}} \alpha^k_\pi \text{ , for all } i,j,k.
    \label{EQ:Birkhoff-q}
  \end{equation}
  From $\q$, define the lottery $\hat\q$ by
  \[
    \hat{q}_{ij}^k = \begin{cases} q_{hj}^k & \text{ if } i=l\\q_{lj}^k & \text{ if } i=h\\q_{ij}^k & \text{ if } i\notin\{l,h\}
    \end{cases}.
  \]
and note that $u_l(\hat\q;h) = u_l(\q;l)$, $u_h(\hat\q;l) = u_h(\q;h)$,
whereas $u_i(\hat\q;j) = u_i(\q;j)$ when
$(i,j) \notin \{(l,l),(l,h),(h,l),(h,h)\}$.

For any permutation $\pi \in S_n$, denote by $\hat{\pi}$ the permutation given by
\[
  \hat\pi(i) = \begin{cases} \pi(h) & \text{ if } i=l\\ \pi(l) & \text{ if } i=h\\ \pi(i) & \text{ if } i \notin \{l,h\}
  \end{cases}.
\]
In other words, $\hat{\pi}$ just swaps the images of $l$ and $h$, respectively, compared to $\pi$. We can now observe that
  \[
    \hat{q}_{ij}^k = p_k \sum_{\substack{\pi \in S_n \\ \pi(i)=j}} \alpha^k_{\hat\pi} \text{ , for all } i,j,k \enspace .
  \]

  Since $(\q,p)$ is an optimal solution of
  LP~(\ref{EQ:MaxWeightedSumLP}) and $(\hat\q,p)$ is a feasible  solution, we have that
  $\sum_{i=1}^n w_i u_i(\q;i) \geq \sum_{i=1}^n w_i
  u_i(\hat\q;i)$. Using Equation~\ref{EQ:Birkhoff-q} we can rewrite the two sides of the inequality as
\[
  \sum_{i=1}^n w_i u_i(\q;i) = \sum_{k=1}^mp_k\sum_{\pi\in S_n}\alpha_\pi^k \sum_{i=1}^n w_i u_{i\pi(i)}^k
\]
and
\[
  \sum_{i=1}^n w_i u_i(\hat\q;i) = \sum_{k=1}^mp_k\sum_{\pi\in S_n}\alpha_{\hat\pi}^k \sum_{i=1}^n w_i u_{i\pi(i)}^k = \sum_{k=1}^mp_k\sum_{\pi\in S_n}\alpha_\pi^k \sum_{i=1}^n w_i u_{i\hat\pi(i)}^k,
\]
respectively. Optimality of $\q$ implies that whenever
$p_k\alpha_\pi^k>0$, we have
\[
  \sum_{i=1}^n w_i u_{i\pi(i)}^k \geq \sum_{i=1}^n w_i u_{i\hat\pi(i)}^k,
\]
or equivalently,
\[
  w_l u_{l\pi(l)}^k + w_h u_{h\pi(h)}^k \geq w_l u_{l\pi(h)}^k + w_h u_{h\pi(l)}^k,
\]
which may be rewritten as
\begin{equation}
  w_h(u_{h\pi(h)}^k - u_{h\pi(l)}^k) \geq w_l(u_{l\pi(h)}^k - u_{l\pi(l)}^k).
  \label{EQ:weight-envy-inequality}
\end{equation}
Suppose for contradiction that $u_{l\pi(h)}^k > u_{l\pi(l)}^k$. Then
since $w_l>0$ we also have that $u_{h\pi(h)}^k > u_{h\pi(l)}^k$. By Definition~\ref{DEF:rho} this
means that $(k,l,\pi(l),h,\pi(h)) \in J$ and thus
\[
  \rho < (u_{l\pi(h)}^k - u_{l\pi(l)}^k)/(u_{h\pi(h)}^k - u_{h\pi(l)}^k).
\]
From Equation~(\ref{EQ:weight-envy-inequality}) we then obtain
$w_h > \rho w_l$, contradicting the assumption that
$w_h \leq \rho w_l$.

Since this holds for $k$ and $\pi$ we can finally conclude that
\[
  u_l(\q;l) = \sum_{k=1}^m p_k\sum_{\pi\in S_n}\alpha_\pi^k u_{l\pi(l)}^k \geq \sum_{k=1}^m p_k\sum_{\pi\in S_n}\alpha_\pi^k u_{l\pi(h)}^k = u_l(\q;h). \qedhere
\]
\end{proof}

\section{Solving optimization and feasibility problems using the linear-OPT-gate}
\label{app:linearOPT}
In this section we give precise statements of the optimization and feasibility problems that are solvable by the linear-OPT-gate of Filos-Ratsikas~et~al.~\cite{Filos-RatsikasH2023-PPAD}.

\subsection{Optimization problems}
The linear-OPT gate is a construction of a piecewise linear arithmetic pseudo-circuit. It is parametrized (meaning that these are fixed constants) by the following:
\begin{itemize}
    \item Numbers $n,m,k \in \NN$.
    \item A (rational) matrix $A \in \RR^{m\times n}$.
    \item A piecewise linear arithmetic circuit $G_{\partial f} \colon \RR^n \times \RR^k \times [0,1]^\ell \to \RR \times [0,1]^\ell$.
\end{itemize}
It takes as input (meaning that these are given as input variables) by the following:
\begin{itemize}
    \item Vectors $b \in \RR^m$ and $c \in \RR^k$.
    \item A number $R \in \RR$.
\end{itemize}
The linear-OPT gate computes an optimal solution of the following optimization problem $\mathcal{C}$ in decision variables $x \in \RR^n$:
\begin{center}\underline{Optimization Program $\mathcal{C}$}\end{center}
\begin{equation}\label{eq:OPT-gate-general}
\begin{split}
\min \quad &f(x;c) \\
\text{ s.t.} \quad & Ax \leq b\\
& x \in [-R,R]^n
\end{split}
\end{equation}
whenever the following conditions hold:
\begin{itemize}
\item The feasible domain $\{x \in [-R,R]^n : Ax \leq b\}$ is not empty.
\item The map $x \mapsto f(x;c)$ is a convex function on the feasible domain and its subgradient is given by the pseudo-circuit $G_{\partial f}$.
\end{itemize}

\subsection{Feasibility problems}
Using the linear-OPT, a piecewise linear arithmetic pseudo-circuit solving feasibility problems with conditional constraints can be constructed. It is parametrized (meaning that these are fixed constants) by the following:
\begin{itemize}
    \item Numbers $n,m,k \in \NN$.
    \item A (rational) matrix $A \in \RR^{m\times n}$.
    \item Piecewise linear circuit arithmetic circuits $h_i \colon \RR^k \to \RR$, for $i=1,\dots,m$.
\end{itemize}
It takes as input (meaning that these are given as input variables) by the following:
\begin{itemize}
    \item Vectors $b \in \RR^m$ and $y \in \RR^k$.
    \item A number $R \in \RR$.
\end{itemize}
The pseudo-circuit outputs a feasible solution of the following feasibility problem $\mathcal{Q}$ in decision variables $x \in \mathbb{R}^n$:
\begin{center}\underline{Feasibility Program $\mathcal{Q}$}\end{center}
\begin{equation}\label{eq:feasibility-general}
\begin{split}
h_i(y) > 0 \implies a_i^\mathsf{T} x \leq b_i\\
x \in [-R,R]^n
\end{split}
\end{equation}
whenever it is feasible. Note also that ordinary inequality constraints are a special case of conditional linear constraints, obtained by setting $h_i(y) = 1$ above.

\end{document}